\documentclass{arxiv}   
\usepackage{ws-toc,ws-multind}
\makeindex{author}                     
\usepackage[latin1]{inputenc}
\usepackage{dsfont}
\newcommand\R{{\ensuremath {\mathbb R} }}
\newcommand\C{{\ensuremath {\mathbb C} }}
\newcommand\N{{\ensuremath {\mathbb N} }}
\newcommand\Z{{\ensuremath {\mathbb Z} }}

\newcommand\1{{\ensuremath {\mathds 1} }}

\newcommand\cZ{\mathcal{Z}}
\renewcommand\phi{\varphi}

\newcommand{\cP}{\mathcal{P}}
\newcommand{\cN}{\mathcal{N}}

\newcommand{\cB}{\mathcal{B}}

\newcommand{\cF}{\mathcal F}

\newcommand{\cM}{\mathcal M}

\newcommand{\cR}{\mathcal R}
\newcommand{\tr}{{\rm tr}\,}

\renewcommand{\epsilon}{\varepsilon}
\newcommand{\qed}{\hfill$\square$}
\newcommand\ii{{\ensuremath {\infty}}}
\newcommand\pscal[1]{{\ensuremath{\left\langle #1 \right\rangle}}}
\newcommand{\norm}[1]{ \left| \! \left| #1 \right| \! \right| }

\begin{document}
%
%
%
%

\wstoc{The Thermodynamic Limit of Quantum Coulomb Systems: A New Approach}{C. Hainzl, M. Lewin and J.P. Solovej}

\title{\large\sf The Thermodynamic Limit of Quantum Coulomb Systems: A New Approach\footnotetext{\textnormal Talk given by M.L. at QMath10, 10th Quantum Mathematics International Conference, Moeciu, Romania, September 10--15, 2007. \copyright\,\textit{2008 by the authors. This work may be reproduced, in its entirety, for non-commercial purposes.}}}

\author{Christian HAINZL}
\aindx{Hainzl, Ch.}
\address{Department of Mathematics, UAB\\ 
Birmingham, AL 35294-1170 - USA\\
Email:\texttt{hainzl@math.uab.edu}}

\author{Mathieu LEWIN}
\address{CNRS and Laboratoire de Mathématiques (UMR 8088)\\
Universit{\'e} de Cergy-Pontoise\\
2, avenue Adolphe Chauvin\\
 95302 Cergy-Pontoise Cedex - France\\
Email: \texttt{Mathieu.Lewin@math.cnrs.fr}}

\author{Jan Philip SOLOVEJ}
\address{University of Copenhagen\\ Department of Mathematics\\
Universitetsparken 5\\
 2100 Copenhagen \O - Denmark\\
Email: \texttt{solovej@math.ku.dk}}

\begin{abstract}
We present recent works \cite{1,2} on the thermodynamic limit of quantum Coulomb systems. We provide a general method which allows to show the existence of the limit for many different systems. 
\end{abstract}

\bigskip

Ordinary matter is composed of electrons (negatively charged) and nuclei (positively charged) interacting via Coulomb forces. The potential between two particles of charges $z$ and $z'$ located at $x$ and $x'$ in $\R^3$ is
$$\frac{zz'}{|x-x'|}.$$
There are two difficulties which occur when trying to describe systems composed of electrons and nuclei. Both have to do with the physical problem of \emph{stability} of quantum systems.

The first is due to the singularity of $1/|x|$ at $0$: it is necessary to explain why a particle will not rush to a particle of the opposite charge. One of the first major triumphs of the theory of quantum mechanics is the explanation it gives of the stability of the hydrogen atom (and the complete description of its spectrum) and of other microscopic quantum Coulomb systems, via the uncertainty principle. 
Stability means that the total energy of the considered system cannot be arbitrarily
negative.  If there was no such lower bound to the energy it would be possible in principle to extract
an infinite amount of energy. One often refers to this kind of stability as {\it stability of the first kind} \cite{Lieb1,Lieb2}.
If we denote by $E(N)$ the ground state energy of the system under consideration, for $N$ particles stability of the first kind can be written
\begin{equation}
 E(N)>-\ii.
\label{stability_first_kind}
\end{equation}

In proving \eqref{stability_first_kind} for Coulomb systems, a major role is played by the uncertainty principle which for nonrelativistic systems is mathematically expressed by  the critical Sobolev embedding $H^1(\mathbb{R}^3)\hookrightarrow L^6(\mathbb{R}^3)$. The latter allows to prove Kato's inequality 
$$\forall\epsilon>0,\qquad\frac1{|x|}\leq \epsilon(-\Delta)+\frac{1}\epsilon,$$
which means that the Coulomb potential is controlled by the kinetic energy.

The second issue concerns the slow decay of $1/|x|$ at infinity and this has to do with the macroscopic behavior of quantum Coulomb systems. 
It is indeed necessary to explain how a very large number of electrons and nuclei can stay bounded together to form macroscopic systems, although each particle interacts with a lot of other charged particles due to the long tail of the Coulomb interaction potential.
Whereas the stability of atoms was an early triumph of quantum mechanics it, surprisingly, took nearly forty years before the question of stability of everyday macroscopic objects was even raised (see Fisher and Ruelle \cite{fr}). A rigorous answer to the question came shortly thereafter in what came to be known as the Theorem on Stability of Matter proved first by Dyson and Lenard \cite{DL1}. 

The main question is how the lowest possible energy $E(N)$ appearing  in \eqref{stability_first_kind} depends on the (macroscopic) number $N$ of particles in the object.
More precisely, one is interested in proving a behavior of the form
\begin{equation}
 E(N)\sim_{N\to\ii}\bar{e}N.
\label{thermo_lim_N_intro}
\end{equation}
This behavior as the number of particles grows is mandatory to explain why matter does not collapse or explode in the thermodynamic limit.
Assume that \eqref{thermo_lim_N_intro} does not hold and that for instance $E(N)\sim_{N\to\ii}cN^p$ with $p\neq1$. Then $|E(2N)-2E(N)|$ becomes very large as $N\gg1$. Depending on $p$ and the sign of the constant $c$, a very large amount of energy will be either released when two identical systems are put together, or necessary to assemble them.
The constant $\bar{e}$ in \eqref{thermo_lim_N_intro} is the energy per particle.

Stability of Matter is itself a necessary first step towards a proof of \eqref{thermo_lim_N_intro} as it can be expressed by the lower bound
\begin{equation}
 E(N)\geq -\kappa N.
\label{stability_of_matter_N_intro}
\end{equation}
Put differently, the lowest possible energy calculated per particle
cannot be arbitrarily negative as the number of particles increases.
This is also often referred to as {\it stability of the second kind} \cite{Lieb1,Lieb2}.

A maybe more intuitive notion of stability would be to ask for the
volume occupied by a macroscopic object (in its ground state). Usually this volume is proportional to the number of particles $N$. 
Denoting by $\Omega$ a domain in $\R^3$ which is occupied by the system under consideration and by $E(\Omega)$ its (lowest possible) energy, \eqref{thermo_lim_N_intro} then reads
\begin{equation}
 E(\Omega)\sim_{|\Omega|\to\ii}\bar{e}|\Omega|
\label{thermo_lim_V_intro}
\end{equation}
where $|\Omega|$ is the volume of $\Omega$. Stability of the second kind is expressed as
\begin{equation}
 E(\Omega)\geq -\kappa|\Omega|.
\label{stability_of_matter_V_intro}
\end{equation}

Instead of the ground state energy, one can similarly consider the free energy $F(\Omega,\beta,\mu)$ at temperature $T=1/\beta$ and chemical potential $\mu$. One is then interested in proving the equivalent of \eqref{thermo_lim_V_intro}
\begin{equation}
 F(\Omega,\beta,\mu)\sim_{|\Omega|\to\ii}\bar{f}(\beta,\mu)|\Omega|
\label{thermo_lim_V_intro2}
\end{equation}
where $\bar{f}(\beta,\mu)$ is the free energy per unit volume. 

Large quantum Coulomb systems have been the object of an important investigation in the last decades and many techniques have been developed.  A result like \eqref{stability_of_matter_N_intro} (or equivalently \eqref{stability_of_matter_V_intro}) was first proved for quantum electrons and nuclei by Dyson and Lenard \cite{DL1}. 
After the original proof by Dyson and Lenard several other proofs were given.  Lieb and Thirring~\cite{LT} in particular presented an elegant and simple proof relying on an {uncertainty principle for fermions}. The different techniques and results concerning stability of matter were reviewed in several articles \cite{Lieb1,Lieb2,l-heisen,Loss,Solovej_rev}.

It is very important that the negatively charged particles (the electrons) are fermions. It was discovered by Dyson \cite{dyson1} that the Pauli exclusion principle is essential for Coulomb systems: charged bosons are alone not stable because their ground state energy satisfies $E(N)\sim -CN^{7/5}$, as was proved later \cite{CLY1,LiSo,Solovej3}.

A result like \eqref{thermo_lim_N_intro} (or equivalently \eqref{thermo_lim_V_intro}) was first proved by Lieb and Lebowitz \cite{LL} for a system containing electrons and nuclei both considered as quantum particles, hence invariant by rotation. Later Fefferman gave a different proof \cite{F} for the case where the nuclei are classical particles placed on a lattice, a system which is not invariant by rotation.

In a recent work \cite{1,2}, we provide a new insight in the study of the thermodynamic limit of quantum systems, by giving a general proof of \eqref{thermo_lim_V_intro} or \eqref{thermo_lim_V_intro2} which can be applied to many different quantum systems including those studied by Lieb and Lebowitz \cite{LL} or Fefferman \cite{F}, and others which were not considered before. Our goal was to identify the main general physical properties of the free energy which are sufficient to prove the existence of the thermodynamic limit. However, for the sake of simplicity we will essentially address the crystal case in this paper and we refer to our works \cite{1,2} for a detailed study of the other cases.

In proving the existence of the thermodynamic limit of Coulomb quantum systems, the most difficult task is to quantify \emph{screening}. Screening means that matter is arranged in such a way that it is essentially locally neutral, hence the electrostatic potential created by any subsystem decays much faster than expected. This effect is the main reason of the stability of ordinary matter but it is very subtle in the framework of quantum mechanics because the particles are by essence delocalized. In our approach, we shall heavily rely on an electrostatic inequality which was proved by Graf and Schenker \cite{GS,G} and which serves as a way to quantify screening. It was itself inspired by previous works of Conlon, Lieb and Yau \cite{CLY1,CLY2}, for systems interacting with the Yukawa potential. Fefferman used a similar idea in his study of the crystal case \cite{F}.

Like in previous works, our method consists in first showing the existence of the limit \eqref{thermo_lim_V_intro2} for a specific domain $\triangle$ which is dilated (and possibly rotated and translated). Usually $\triangle$ is chosen to be a ball, a cube or a tetrahedron. In the applications \cite{2} we always choose a tetrahedron as we shall use the Graf-Schenker inequality \cite{GS} which holds for this type of domains. 
The second step consists in showing the existence of the limit \eqref{thermo_lim_V_intro2} for any (reasonable) sequence of domains $\{\Omega_n\}$ such that $|\Omega_n|\to\ii$. This is important as in principle the limit could depend on the chosen sequence, a fact that we want to exclude for our systems. We shall specify later what a ``reasonable'' sequence is. Essentially some properties will be needed to ensure that boundary effects always stay negligible. 

It is to be noticed that our method (relying on the Graf-Schenker inequality) is primarily devoted to the study of quantum systems interacting through Coulomb forces. It might be applicable to other interactions but we shall not address this question here.

Proving a result like \eqref{thermo_lim_V_intro} or \eqref{thermo_lim_V_intro2} is only a first step in the study of the thermodynamic limit of Coulomb quantum systems. An interesting open problem is to prove the convergence of \emph{states} (or for instance of all $k$-body density matrices) and not only of energy levels. For the crystal case, convergence of the charge density or of the first order density matrix was proved for simplified models from Density Functional Theory or from Hartree-Fock theory \cite{LS,CLL1}. A result of this type was also proved for the Hartree-Fock approximation of no-photon Quantum Electrodynamics \cite{HLSo}. 

Another (related) open question is to determine the next order in the asymptotics of the energy in the presence of local perturbations. Assume for instance that the crystal possesses a local defect modelled by a local potential $V$ and denote the ground state energy in the domain $\Omega$ by $E^V(\Omega)$. Since $V$ is local, it does not contribute to the energy in the first order of the thermodynamic limit. One is then interested in proving a behavior like $E^V(\Omega)=E^0(\Omega)+f(V)+o(1)_{|\Omega|\to\ii}$. Such a result was recently proved for the \emph{reduced} Hartree-Fock model of the crystal with the exchange term neglected \cite{CDL}. This includes an identification of the function $f(V)$. This program was also tackled for the Hartree-Fock model (with exchange term) of no-photon Quantum Electrodynamics \cite{HLSo}.

The present paper is organized as follows. In the first section we introduce the model for the crystal and state our main theorem. In Section 2, we briefly describe two other quantum systems which we can treat using our method. Section 3 is devoted to the presentation of our new approach, in a quite general setting, together with hints on how it can be applied to the crystal case.

\section{The Crystal Case}\label{sec:def_crystal}
For simplicity, we put identical nuclei of charge $+1$ on each site of $\Z^3$. The results below can be generalized to any periodic system. Let $\Omega$ be a bounded open set of $\R^3$ and define the $N$-body Hamiltonian in $\Omega$ by
$$H^N_\Omega:=\sum_{i=1}^N-\frac{\Delta_{x_i}}{2}+V_\Omega(x_1,...,x_N),$$
where
$$V_\Omega(x)=\sum_{i=1}^N\sum_{R\in\Z^3\cap\Omega}\frac{-1}{|R-x_i|}+\frac12\sum_{1\leq i\neq j\leq N}\frac{1}{|x_i-x_j|}+\frac12\sum_{R\neq R' \in\Z^3\cap\Omega}\frac{1}{|R-R'|}.$$ 
Here $-\Delta$ is the \emph{Dirichlet Laplacian} on $\Omega$ (we could as well consider another boundary condition). The Hamiltonian $H^N_\Omega$ acts on $N$-body fermionic wavefunctions $\Psi(x_1,..,x_N)\in\bigwedge_1^N L^2(\Omega)$.
Stability of the first kind states that the spectrum of $H^N_\Omega$ is bounded from below:
$$E_\Omega^N=\inf_{\substack{\Psi\in\bigwedge_1^N H^1_0(\Omega),\\ \norm{\Psi}_{L^2}=1}}\pscal{\Psi,H^N_\Omega\Psi}=\inf\sigma_{\bigwedge_1^N L^2(\Omega)}(H^N_\Omega)>-\ii.$$
We may define the ground state energy in $\Omega$ by
\begin{equation}
E(\Omega):=\inf_{N\geq0}E_\Omega^N. 
\label{def_GS_energy}
\end{equation}

It is more convenient to express \eqref{def_GS_energy} in a grand canonical formalism. We define the (electronic) Fock space as
$$\cF_\Omega:=\C\oplus\bigoplus_{N\geq1}\bigwedge_1^N L^2(\Omega)$$
The grand canonical Hamiltonian is then given by $H_\Omega:=\bigoplus_{N\geq0} H^N_\Omega$ with the convention that $H^0_\Omega=(1/2)\sum_{R\neq R' \in\Z^3\cap\Omega}|R-R'|^{-1}\in\C$. The number operator reads $\cN:=\bigoplus_{N\geq0} N$. It is then straightforward to check that
$$E(\Omega) =  \inf\sigma_{\cF_\Omega}(H_\Omega) =  \inf_{\substack{\Gamma\in\cB(\cF_\Omega),\ \Gamma^*=\Gamma,\\ 0\leq \Gamma\leq 1,\ \tr_{\cF_\Omega}(\Gamma)=1.}}\tr_{\cF_\Omega}\left(H_\Omega\Gamma\right).$$
The free energy at temperature $1/\beta$ and chemical potential $\mu\in\R$ is defined by
\begin{eqnarray}
 F(\Omega,\beta,\mu) & := &  \inf_{\substack{\Gamma\in\cB(\cF_\Omega),\ \Gamma^*=\Gamma,\\ 0\leq \Gamma\leq 1,\ \tr_{\cF_\Omega}(\Gamma)=1.}}\left(\tr_{\cF_\Omega}((H_\Omega-\mu\cN)\Gamma)+\frac{1}{\beta}\tr_{\cF_\Omega}(\Gamma\log\Gamma)\right)\nonumber\\
 & = & -\frac1\beta \log\tr_{\cF_\Omega}\left[e^{-\beta(H_\Omega-\mu \cN)}\right].\label{def_Free_energy}
\end{eqnarray}
As explained in Introduction, our purpose is to prove that 
\begin{equation}
 E(\Omega)\sim_{|\Omega|\to\ii}\bar e|\Omega|\quad\text{and}\quad F(\Omega,\beta,\mu)\sim_{|\Omega|\to\ii}\bar f(\beta,\mu)|\Omega|
\label{goal}
\end{equation}
in an appropriate sense. The first important property of $E$ and $F$ is the stability of matter.
\begin{theorem}[Stability of Matter \cite{2}]\label{thm:stability}
There exists a constant $C$ such that the following holds:
$$E(\Omega)\geq -C|\Omega|,\qquad F(\Omega,\beta,\mu)\geq -C\left(1+\beta^{-5/2}+\max(0,\mu)^{5/2}\right)|\Omega|$$
for any bounded open set $\Omega\subset\R^3$ and any $\beta>0$, $\mu\in\R$.
\end{theorem}
\begin{proof}[Sketch of the proof]
The first step is to use an inequality for classical systems due to Baxter \cite{Baxter}, improved later by Lieb and Yau \cite{LY}, and which allows to bound the full $N$-body Coulomb potential by a one-body potential:
\begin{equation}
 V(x_1,...,x_N) \geq -\sum_{i=1}^N\frac{3/2+\sqrt2}{\delta(x_i)}
\label{ineg_Baxter}
\end{equation}
where $\delta(x)=\inf_{R\in\Z^3}|x-R|$ is the distance to the closest nucleus. Hence we have the lower bound
$$H^N_\Omega\geq \sum_{i=1}^N\left(-\frac{\Delta_{x_i}}{2}-\frac{3/2+\sqrt2}{\delta(x_i)}\right).$$
Next we split the kinetic energy in two parts and we use the uncertainty principle to show that on $L^2(\Omega)$
$$-\frac{\Delta}{4}-\frac{3/2+\sqrt2}{\delta(x)}\geq -C.$$
In proving this lower bound, one uses the Sobolev inequality in a small ball around each nucleus, exploiting the fact that the nuclei are fixed and separated by a distance at least one to each other. The proof of the stability of matter for systems with classical nuclei whose position is unknown is more difficult and it uses the improved version of \eqref{ineg_Baxter} contained in the paper by Lieb and Yau \cite{LY}, as explained in our work \cite{2}. This shows
\begin{equation}
 H^N_\Omega\geq \sum_{i=1}^N\left(-\frac{\Delta_{x_i}}{4}-C\right)\quad\text{hence}\quad H_\Omega\geq -\frac14\sum_i\Delta_i-C\cN
\label{estim_Hamil}
\end{equation}
on $L^2(\Omega)$ and $\cF_\Omega$ respectively.
The last step is to use the Lieb-Thirring inequality \cite{LT} which states that
\begin{equation}
 \pscal{\sum_{i=1}^N\left(-\Delta_{x_i}\right)\Psi,\Psi}\geq C_{\rm LT}\int_{\Omega}\rho_\Psi(x)^{5/3}dx
\label{eq_LT}
\end{equation}
for all $N\geq1$ and all $N$-body fermionic wavefunction $\Psi\in \bigwedge_1^N L^2(\Omega)$. The density of charge $\rho_\Psi$ is as usual defined by $\rho_\Psi(x)=N\int_{\Omega^{N-1}}|\Psi(x,y)|^2dy$. Using the fact that $\int_\Omega\rho_\Psi=N$ and Hölder's inequality, \eqref{eq_LT} yields on the Fock space $\cF_\Omega$
\begin{equation}
\sum_i(-\Delta_{x_i})\geq C_{\rm LT}|\Omega|^{-2/3}\cN^{5/3}.
\label{eq_LT2}
\end{equation}
Hence we obtain $H_\Omega\geq (C_{\rm LT}/4)|\Omega|^{-2/3}\cN^{5/3}-C\cN$ which, when optimized over $N$, gives the result for the ground state energy. 

For the free energy, we use \eqref{estim_Hamil}, \eqref{eq_LT2} and Peierls' inequality \cite{Ruelle,Wehrl} to get
$$F(\beta,\mu,\Omega)\geq-\frac1\beta\log\tr_{\cF}\left(e^{-\beta\sum_i(-\Delta_i)/4}\right)-C(1+\mu_+^{5/2})|\Omega|.$$
The first term of the r.h.s. is the free energy of a free-electron gas which is bounded below by $-C(1+\beta^{-5/2})|\Omega|$ in the thermodynamic limit \cite{2}.
\end{proof}

In order to state our main result, we need the
\begin{definition}[Regular sets in $\R^3$]
Let be $a>0$ and $\epsilon>0$.
We say that a bounded open set $\Omega\subseteq\R^3$ has \emph{an $a$-regular boundary in the sense of Fisher} if, denoting by $\partial\Omega=\overline{\Omega}\setminus{\Omega}$ the boundary of ${\Omega}$,
\begin{equation}
 \forall t\in[0,1],\qquad \left|\left\{x\in\R^3\ |\ \textnormal{d}(x,\partial\Omega)\leq |\Omega|^{1/3}t\right\}\right|\leq |\Omega|\,a\, t.
\label{def_reg_boundary}
\end{equation}
We say that a bounded open set $\Omega\subseteq\R^3$ satisfies the \emph{$\varepsilon$-cone property} if for any $x\in\Omega$ there is a unit vector $a_x\in\R^3$ such that 
$$\{y\in\R^3\ |\  (x-y)\cdot a_x> (1-\varepsilon^2) |x-y|,\ |x-y|<\varepsilon\} \subseteq\Omega.$$
We denote by $\cR_{a,\varepsilon}$ the set of all $\Omega\subseteq\R^3$ which have an $a$-regular boundary and such that both $\Omega$ and $\R^3\setminus\Omega$ satisfy the $\varepsilon$-cone property.
\end{definition}

Note that any open convex set is in $\cR_{a,\varepsilon}$ for some $a>0$ large enough and $\varepsilon>0$ small enough  \cite{1}.
We may state our main 
\begin{theorem}[Thermodynamic Limit for the Crystal \cite{2}]\label{thm:crystal}
There exist $\bar e\in\R$ and a function $\bar f:(0,\ii)\times\R\to\R$ such that the following holds: for any sequence $\{\Omega_n\}_{n\geq1}\subseteq \cR_{a,\epsilon}$ of domains with $|\Omega_n|\to\ii$, $|\Omega_n|^{-1/3}{\rm diam}(\Omega_n)\leq C$, $a\geq a_0>0$ and $0<\varepsilon\leq\varepsilon_0$
\begin{equation}
\lim_{n\to\ii}\frac{E(\Omega_n)}{|\Omega_n|}=\bar e,\qquad \lim_{n\to\ii}\frac{F(\Omega_n,\beta,\mu)}{|\Omega_n|}=\bar f(\beta,\mu).
\label{limit_FE}
\end{equation}
Moreover $\bar f$ takes the form $\bar f(\beta,\mu)=\phi(\beta)-\mu$.
\end{theorem}

\begin{remark}
We know from \cite[Appendix A p. 385]{LL} and \cite[Lemma 1]{Fisher} that if each set $\Omega_n$ of the considered sequence is connected, then automatically $|\Omega_n|^{-1/3}{\rm diam}(\Omega_n)\leq C$.
\end{remark}

A very similar result was proved by C. Fefferman \cite{F}. Our result is more general: we allow any sequence $\Omega_n$ tending to infinity and which is regular in the sense that $\{\Omega_n\}_{n\geq1}\subseteq \cR_{a,\epsilon}$. In Fefferman's paper \cite{F}, $\Omega_n=\ell_n(\Omega+x_n)$ where $\ell_n\to\ii$, $\Omega$ is a fixed convex open set and $x_n$ is any sequence in $\R^3$. These sets are always in $\cR_{a,\epsilon}$ for some $a,\varepsilon>0$.

In our work \cite{2} a result even more general than Theorem \ref{thm:crystal} is shown: we are able to prove the existence of the same thermodynamic limit if the crystal is \emph{locally perturbed} (for instance finitely many nuclei are moved or their charge is changed). A similar result can also be proved for the Hartree-Fock model.

\section{Other models}\label{sec:examples}
Our approach \cite{1,2} is general and it can be applied to a variety of models, not only the crystal case. We quickly mention two such examples. It is interesting to note that for these other models, we do not need the cone property and we can weaken the assumptions on the regularity of the boundary by replacing $t$ on the r.h.s. of \eqref{def_reg_boundary} by any $t^p$, $0<p\leq1$. Details may be found in our article \cite{2}. Roughly speaking, when the system is ``rigid'' like for the crystal (the nuclei are fixed), the proof is more complicated and more assumptions are needed on the sequence of domains to avoid undesirable boundary effects.

\subsection{Quantum particles in a periodic magnetic field.} Define the magnetic kinetic energy $T(A)=(-i\nabla +A(x))^2$ where $B=\nabla\times A$ is periodic (for instance constant) and $A\in L^2_{\rm loc}(\R^3)$. Next, consider the Hamiltonian
$$H_\Omega^{N,K}:=\sum_{i=1}^NT(A)_{x_i}+\sum_{k=1}^KT(A)_{R_k}+V(x,R),$$
$$V(x,R)=\sum_{i,k}\frac{-z}{|R_k-x_i|}+\frac12\sum_{i\neq j}\frac{1}{|x_i-x_j|}
+\frac12\sum_{k\neq k'}\frac{z^2}{|R_k-R_{k'}|}$$
The ground state energy is this time defined as
$$E'(\Omega):=\inf_{N,K\geq0}\inf\sigma_{\bigwedge_1^NL^2(\Omega)\otimes S\!\bigotimes_{1}^KL^2(\Omega)}\left(H^{N,K}_\Omega\right).$$
We do not precise the symmetry $S$ of the particles of charge $z$ which can be bosons or fermions.
A formula similar to \eqref{def_Free_energy} may be used for the free energy on the (electronic and nucleic) Fock space. We prove in our paper \cite{2} a result similar to Theorem \ref{thm:crystal} for this model.
Lieb and Lebowitz already proved it in the seminal paper \cite{LL} when $A\equiv0$. They used as an essential tool the rotation-invariance of the system to obtain screening. When $A\neq0$ the system is no more invariant by rotations and their method cannot be applied.

\subsection{Classical nuclei with optimized position.} For all $R\subset\Omega$, $\#R<\ii$, let us define
$$H^{N,R}_\Omega:=\sum_{i=1}^N-\frac{\Delta_{x_i}}{2}+V(x,R)$$
and the associated ground state energy by
$$E''(\Omega):=\inf_{N\geq0}\inf_{\substack{R\subset\Omega,\\ \#R<\ii}}\inf\sigma_{\bigwedge_1^NL^2(\R^3)}\left(H^{N,R}_\Omega\right).$$
We could as well optimize the charges in $[0,z]$ of the nuclei without changing the energy \cite{DL,2}. However, the free energy itself is not the same when the charges of the nuclei are optimized or not \cite{2}.

Surprisingly, to our knowledge the existence of the thermodynamic limit for this model was unknown. A result similar to Theorem \ref{thm:crystal} is proved in our paper \cite{2} for $E''$.

\section{A general method}
In this section, we give the main ideas of our new approach which allows to prove Theorem \ref{thm:crystal} and its counterparts for the other models quoted before.

\subsection{Screening via the Graf-Schenker inequality}
\begin{figure}[t]
\centering
\includegraphics[width=11cm]{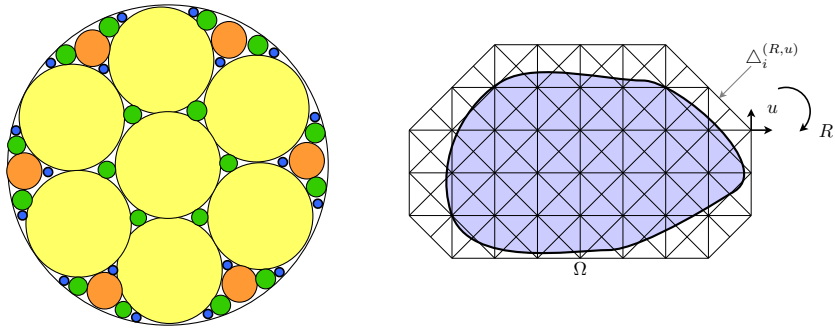}
\caption{A comparison between the original method of Lieb and Lebowitz \cite{LL} (left) and our method based on the Graf-Schenker inequality \cite{GS,1,2} (right).}
\label{fig:compare}
\end{figure}
As mentioned in the introduction, an important step is to quantify screening. For quantum nuclei without a magnetic field ($A\equiv0$), Lieb and Lebowitz used \cite{LL} the following method (see Figure \ref{fig:compare}). First they took a big ball $B$ which they packed with several small balls $B_k$ of different size. In each of these balls, they took the (neutral) ground state of the corresponding ball. As the system is invariant under rotations, they can freely rotate each ground state. Averaging over rotations of all the small balls, they reduced the computation of the interaction between them to that of classical pointwise particles located at the center of the balls, by Newton's theorem. As each subsystem is neutral, this interaction vanishes. This proves \emph{a fortiori} that there exists an adequate rotation of each system in each little ball such that the total interaction between them cancels. Choosing this configuration, they could build a test function whose energy is just the sum of the small energies, proving an estimate of the form $E(B)\leq \sum_k E(B_k)$. This inequality can be used to prove the limit for balls. Clearly this trick can only be used for rotation-invariant systems.

Note in the Lieb-Lebowitz proof, a domain (the big ball) is split in several \emph{fixed subdomains} and an average is done \emph{over rotations of the states} in each small domain. This yields an \emph{upper bound} to the energy. The Graf-Schenker inequality is kind of dual to the above method (see Figure \ref{fig:compare}). This time a domain $\Omega$ is split in several subdomains by using a tiling of the space $\R^3$. But the system is frozen in the state of the big domain $\Omega$ and the average is done \emph{over the position of the tiling}. This yields a \emph{lower bound} to the energy of the form $E(\Omega)\geq \sum_k E(\Delta_i^{(r,u)}\cap\Omega)+\text{errors}$, where $\Delta_i^{(r,u)}$ are the tetrahedrons which make up the (translated and rotated) tiling.

The Graf-Schenker inequality was inspired by previous works of Conlon, Lieb and Yau \cite{CLY1,CLY2}. It is an estimate on the Coulomb energy of classical particles. The proof of Fefferman in the crystal case \cite{F} was also based on a \emph{lower bound} on the free energy in a big set and an average over translations of a covering of this set (the method was reexplained later in details by Hugues \cite{Hugues}). Fefferman \cite{F} uses a covering with balls and cubes of different size. The lower bound depends on the number of balls contained in the big domain and of the form of the kinetic energy which is used to control error terms.

Let $G=\R^3\rtimes SO_3(\R)$ be the group of translations and rotations acting on $\R^3$, and denote by $d\lambda(g)$ its Haar measure. 
\begin{lemma}[Graf-Schenker inequality \cite{GS}] 
 Let $\triangle$ be a simplex in $\R^3$. There exists a constant $C$ such that for any $N\in\N$, $z_1,...,z_N\in\R$, $x_i\in\R^3$ and any $\ell>0$,
\begin{equation}
\sum_{1\leq i<j\leq N} \frac{z_iz_j}{|x_i-x_j|}\geq \int_{G} \frac{d\lambda(g)}{|\ell\triangle|}\sum_{1\leq i<j\leq N}\frac{z_iz_j\1_{g\ell\triangle}(x_i)\1_{g\ell\triangle}(x_j)}{|x_i-x_j|}
 -\frac{C}{\ell}\sum_{i=1}^Nz_i^2.
\label{GS_ineq}
\end{equation}
\end{lemma}

In the previous theorem it is not assumed that $\triangle$ yields a tiling of $\R^3$. Up to an error which scales like $\ell$, \eqref{GS_ineq} says that the total Coulomb energy can be bounded from below by the Coulomb energy (per unit volume) of the particles which are in the (dilated) simplex $g\ell\triangle$, averaged over all translations and rotations $g$ of this simplex.

Because of the above inequality, simplices play a specific role in the study of Coulomb systems. Hence proving the existence of the thermodynamic limit for simplices first is natural (as it was natural to consider balls in the Lieb-Lebowitz case due to the invariance by rotation). In the next section we give an abstract setting for proving the existence of the limit when an inequality of the form \eqref{GS_ineq} holds true.

\subsection{An abstract result}
In this section we consider an abstract energy $E:\Omega\in\cM\mapsto E(\Omega)\in\R$ defined on the set $\cM$ of all bounded open subsets of $\R^3$ and we give sufficient conditions for the existence of the thermodynamic limit. In the application, $E$ will be either the ground state energy, or the free energy of the system under consideration. 

We fix a reference set $\triangle\in\cR_{a,\epsilon}$ which is only assumed to be a bounded open convex set in $\R^3$ (it need not be a simplex for this section), such that $0\in\triangle$. Here $a,\epsilon>0$ are fixed. We assume that the energy $E$ satisfies the following five assumptions:

\medskip

\noindent\textbf{(A1)} {\it (Normalization).} $E(\emptyset)=0$.

\medskip

\noindent\textbf{(A2)} {\it (Stability).} $\forall \Omega\in\cM$, $E(\Omega)\geq -\kappa|\Omega|$.

\medskip

\noindent \textbf{(A3)} {\it (Translation Invariance).} $\forall \Omega\in\cR_{a,\epsilon}$, $\forall z\in\Z^3$, $E(\Omega+z)=E(\Omega)$.

\medskip

\noindent \textbf{(A4)} {\it (Continuity).} $\forall\Omega\in\cR_{a,\epsilon}, \Omega'\in\cR_{a',\epsilon'}$ with $\Omega'\subseteq\Omega$ and $\text{d}(\partial\Omega,\partial\Omega')>\delta$,
$$E(\Omega)\leq E(\Omega')+\kappa|\Omega\setminus\Omega'|+|\Omega|\alpha(|\Omega|).$$

\medskip

\noindent \textbf{(A5)} {\it (Subaverage Property).} For all $\Omega\in\cM$, we have
\begin{equation}
E(\Omega)\geq \frac{1-\alpha(\ell)}{|\ell\triangle|}\int_{G}E\big(\Omega\cap g\cdot(\ell\triangle)\big)\,d\lambda(g)-|\Omega|_{\rm r}\,\alpha(\ell)
\label{sliding_equation}
\end{equation}
where $|\Omega|_{\rm r}:=\inf\{|\tilde\Omega|,\ \Omega\subseteq\tilde\Omega,\ \tilde\Omega\in\cR_{a,\epsilon}\}$ is the  regularized volume of $\Omega$.

\bigskip

In the assumptions above $\alpha$ is a fixed function which tends to 0 at infinity and $\delta,a',\epsilon'$ are fixed positive constants. In our work \cite{1}, an even more general setting is provided. First \textbf{(A3)} can be replaced by a much weaker assumption but we do not detail this here. Also a generic class of regular sets $\cR$ is considered instead of $\cR_{a,\epsilon}$. This is because for instance the cone property is only needed for the crystal case and it is not at all necessary in other models, hence the concept of regularity depends on the application.

Notice \textbf{(A4)} essentially says that a small decrease of $\Omega$ will not deacrease too much the energy. A similar property was used and proved in the crystal case by Fefferman \cite[Lemma 2]{F}. Taking $\Omega'=\emptyset$ and using \textbf{(A1)}, property \textbf{(A4)} in particular implies that for any regular set $\Omega\in\cR_{a,\epsilon}$,  $E(\Omega)\leq C|\Omega|$. However this upper bound need not be true for all $\Omega\in\cM$. 
We give a sketch of the proof of the following result in Section \ref{sec:proof}.
\begin{theorem}[Abstract Thermodynamic Limit for $\triangle$ \cite{1}]\label{thm:limit_simplex}
Assume $E:\cM\to\R$ satisfies the above properties \textbf{(A1)--(A5)} for some open convex set $\triangle\in\cR_{a,\epsilon}$ with $0\in\triangle$. There exists $\bar e\in\R$ such that $e_\ell(g)=|\ell\triangle|^{-1}E\big(g\ell\triangle\big)$ converges uniformly towards $\bar e$ for $g\in G=\R^3\rtimes SO(3)$ and as $\ell\to\ii$. 
Additionally, the limit $\bar e$ does not depend on the set $\triangle$\footnote{This means if all the assumptions are true for another set $\triangle'$ then one must have $\bar e'=\bar e$.}.
\end{theorem}

\subsection{Idea of the proof of \textbf{(A1)}--\textbf{(A5)} for the crystal}
Before switching to the abstract case of a general sequence $\{\Omega_n\}$, we give an idea of the proof of \textbf{(A1)}--\textbf{(A5)} in the crystal case. We apply the theory of the previous section to both the ground state energy and the free energy of the crystal which were defined in Section \ref{sec:def_crystal}. First \textbf{(A1)} and \textbf{(A3)} are obvious. Property \textbf{(A2)} is the stability of matter as stated in Theorem \ref{thm:stability}. On the other hand \textbf{(A5)} is essentially the Graf-Schenker inequality \eqref{GS_ineq}, up to some localization issues of the kinetic energy which essentially have already been delt with by Graf and Schenker \cite{GS}.

\begin{figure}
\centering
\includegraphics[width=7cm]{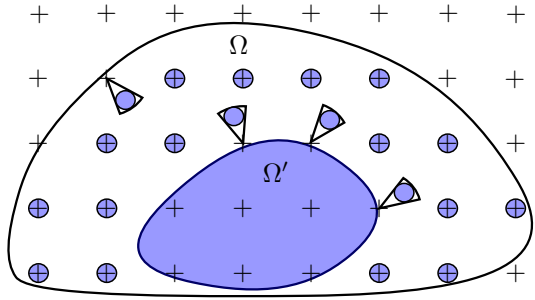}
\caption{Idea of the proof of \textbf{(A4)} for the crystal.}
\label{fig:proof_A4}
\end{figure}

For the crystal the most difficult property is \textbf{(A4)}. The difficulty arises from the fact that this is a very rigid system. For the two other examples mentioned in Section \ref{sec:examples}, \textbf{(A4)} is obvious, the energy being nonincreasing: $E(\Omega)\leq E(\Omega')$. This is because we can simply choose a ground state of $\Omega'$ as a test for $\Omega$ and take the vacuum in $\Omega\setminus\Omega'$. In the crystal case we always have nuclei in $\Omega\setminus\Omega'$ and if we do not put any electron to screen them, they will create an enormous electrostatic energy.

The idea of the proof of \textbf{(A4)} for the crystal is displayed in Figure \ref{fig:proof_A4}. We build a test state in $\Omega$ by considering the ground state in $\Omega'$, and placing one radial electron in a ball of fixed size on top of each nucleus ouside $\Omega'$. By Newton's theorem, the electrostatic potential out of the support of the electron will vanish, hence the energy will simply be $E(\Omega')$ plus the sum of the kinetic energies of the electrons, which is bounded above by a constant times $|\Omega\setminus\Omega'|$ for regular domains. The only problem is that we cannot put an electron on top of the nuclei which are too close to the boundary of $\Omega$ or of $\Omega'$. For these nuclei, using the cone property we can place the ball aside and create a dipole. The difficult task is then to compute a bound on the total interaction between the dipoles and the ground state in $\Omega'$. We prove \cite{2} that it is $o(|\Omega|)$, using a specific version of stability of matter.

\subsection{General domains and strong subadditivity of entropy}
In the previous two subsections, we have presented our abstract theory giving the thermodynamic limit of special sequences built upon the reference set $\triangle$, and we have explained how to apply it to the crystal case. 
For all regular domain sequences we can only get from \textbf{(A5)} a bound of the form
$$\liminf_{n\to\infty}\frac{E(\Omega_n)}{|\Omega_n|}\geq \bar e.$$
In order to get the upper bound, we use a big simplex $L_n\triangle$ of the same size as $\Omega_n$ and a tiling made with simplices of size $\ell_n\ll L_n$, as shown in Figure \ref{fig:general}. We use the ground state of the big simplex $L_n\triangle$ to build a test state in $\Omega_n$, hence giving the appropriate upper bound. To this end, we need some localization features, hence more assumptions in the general theory.

\begin{figure}
\centering
\includegraphics[width=7cm]{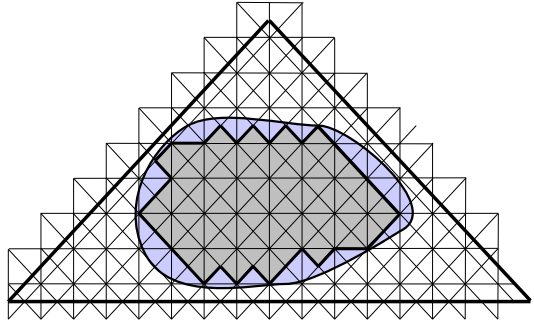}
\caption{Proof for general sequences $\{\Omega_n\}$.}
\label{fig:general}
\end{figure}

It is sufficient \cite{1} to assume that

\smallskip

\noindent $(i)$ $\triangle$ can be used to build a tiling of $\R^3$;

\smallskip

\noindent $(ii)$ the free energy is essentially ``two-body''\footnote{We could as well assume that the energy is $k$-body with $k<\ii$ but this would complicate the assumptions further more.} such that we may write the total energy $E(L_n\triangle)$ as the sum of the energies of the small sets of the tiling, plus the interaction between them and the relative entropy;

\smallskip

\noindent $(iii)$ the entropy is strongly subadditive.

\smallskip

\noindent This is summarized in the following assumption.
We assume that $\Gamma$ is a subgroup of $G$ yielding a tiling of $\R^3$ by means of $\triangle$, i.e. $\overline{\cup_{\mu\in\Gamma} \mu\triangle} =\R^3$ and $\mu\triangle\cap\nu\triangle=\emptyset$ for $\mu\ne\nu$.

\medskip

\noindent \textbf{(A6)} \textsl{(Two-body decomposition).} For all $L$ and $\ell$ we can find $g\in G$ and maps $E_g:\Gamma\to\R$, $I_g:\Gamma\times \Gamma\to\R$, $s_g:\{\cP: \cP\subseteq\Gamma\}\to\R$ such that 

\smallskip

\noindent$\bullet$\; $E_g(\mu)=I_g(\mu,\nu)=0$ if $\ell g\mu\triangle\cap (L\triangle)=\emptyset$;

\smallskip

\noindent$\bullet$\; $\displaystyle E(L\triangle)\geq 
\sum_{\mu\in\Gamma}E_g(\mu)+
\frac{1}{2}\sum_{\mu,\nu\in\Gamma\atop\mu\ne\nu}
I_g(\mu ,\nu) -s_g(\Gamma)-|L\triangle|\alpha(\ell)$;

\smallskip

\noindent$\bullet$\; For all $\cP\subseteq\Gamma$ and $A_\cP=L\triangle\cap\bigcup_{\mu\in\cP}\ell g\mu \triangle$
$$
E(A_\cP)\leq 
\sum_{\mu\in\cP}E_g(\mu)+
\frac{1}{2}\sum_{\mu,\nu\in\cP\atop \mu\ne\nu}
I_g(\mu ,\nu)
 -s_g(\cP)+|A_\cP|\alpha(\ell);
$$

\smallskip

\noindent$\bullet$\; \textsl{(Strong subadditivity).}
for any disjoint subsets $\cP_1$, $\cP_2$, $\cP_3\subseteq\Gamma$ 
$$
s_g(\cP_1\cup\cP_2\cup\cP_3)+s_g(\cP_2)\leq
s_g(\cP_1\cup\cP_2)+s_g(\cP_2\cup\cP_3);
$$

\smallskip

\noindent$\bullet$\; \textsl{(Subaverage property).}
$\displaystyle \int_{G/\Gamma}dg\sum_{\mu,\nu\in\Gamma \atop \mu\ne\nu}I_g(\mu,\nu)
\geq -|L\triangle|\alpha(\ell).$

\medskip

In the applications\footnote{Due to some localization issues of the kinetic energy, it is often needed that the sets of the tiling slightly overlap. See \cite{1} for a generalization in this direction.} the previous quantities are interpreted as follows: $E_g(\cP)$ is the free energy in the union $A_\cP=(L\triangle)\cap \cup_{\mu\in\cP}\ell g\mu \triangle$, $I_g(\mu ,\nu)$ is the interaction energy between the simplices $\ell g\mu \triangle$ and $\ell g\nu \triangle$, and $s_g(\cP)$ is the difference between the entropy of $A_\cP$ and the sum of the entropies of $\ell g\mu \triangle$ with $\mu\in\cP$. 

Conjectured by Lanford and Robinson \cite{LR} the strong subadditivity (SSA) of the entropy in the quantum mechanical case was proved by Lieb and Ruskai\cite{LR1,LR2}. The fact that SSA is very important in the thermodynamic limit was remarked by Robinson and Ruelle\cite{RR} and others\cite{Wehrl}.
In the article\cite{1} we prove the following 
\begin{theorem}[Abstract Limit for general domains \cite{1}]\label{thm:limit_general}
Assume $E:\cM\to\R$ satisfies the properties \textbf{(A1)--(A6)} for some open convex polyhedron $\triangle\in\cR_{a,\epsilon}$ with $0\in\triangle$, yielding a tiling of $\R^3$. Then we have for all sequences $\{\Omega_n\}\subset\cR_{a,\epsilon}$ with $|\Omega_n|\to\ii$ and $|\Omega_n|^{-1/3}{\rm diam}(\Omega_n)\leq C$,
$$\lim_{n\to\ii}\frac{E(\Omega_n)}{|\Omega_n|}=\bar e$$
where $\bar e$ is the limit obtained in Theorem \ref{thm:limit_simplex}.
\end{theorem}

The proof of Theorem \ref{thm:limit_general} is based on a careful estimate of the energy and the interaction energies of boundary terms, ie. of the sets $\ell g\mu \triangle$ which intersect the boundary of the big set $L\triangle$. The application to the crystal is not much more difficult than for Theorem \ref{thm:limit_simplex}. Indeed in the paper of Graf and Schenker \cite{GS}, \eqref{GS_ineq} was expressed using a tiling of $\R^3$ and the last subaverage property of \textbf{(A6)} essentially follows from their ideas \cite{GS}. Strong subadditivity of the entropy is usually expressed via partial traces. A generalization in the setting of localization in Fock space is detailed in our article \cite{2}.

\subsection{Proof of Theorem \ref{thm:limit_simplex}}\label{sec:proof}
Denote as in the Theorem $e_\ell(g)=E(g\ell\triangle)|\ell\triangle|^{-1}$. Notice that \textbf{(A2)}, \textbf{(A4)} with $\Omega'=\emptyset$, and \textbf{(A1)} imply that $e_\ell$ is uniformly bounded on $G$. Also we have by \textbf{(A3)} $e_\ell(u+z,R)=e_\ell(u,R)$ for all $(u,R)\in\R^3\times SO_3(\R)$, $z\in\Z^3$, i.e. $e_\ell$ is periodic with respect to translations. Hence it suffices to prove the theorem for $g=(u,R)\in [0,1]^3\times SO_3(\R)$.

Next we take $\bar g\in G$, $L\gg\ell$ and apply \textbf{(A5)} with $\Omega=\bar g L\triangle$. We get
$$e_L(\bar g)\geq \frac{1-\alpha(\ell)}{|L\triangle|}\int_G \frac{E(\bar g L\triangle\cap g\ell\triangle)}{|\ell\triangle|}dg-\alpha(\ell).$$
Let us introduce the set $\cZ$ of points $z\in\Z^3$ such that $R\ell\triangle+u+z\subset \bar g L\triangle$ for all $u\in[0,1]^3$ and all $R\in SO_3(\R)$. We also define $\partial\cZ$ as the set of points $z\in\Z^3$ such that $(R\ell\triangle+u+z)\cap \bar g L\triangle\neq\emptyset$ for some $(u,R)\in[0,1]^3\times SO_3(\R)$ but $z\notin \cZ$. We obtain using \textbf{(A1)} and \textbf{(A3)}
\begin{eqnarray*}
\int_G \frac{E(\bar g L\triangle\cap g\ell\triangle)}{|\ell\triangle|}dg & = & \sum_{z\in\Z^3}\int_{[0,1]^3}du\int_{SO_3(\R)}dR \frac{E(\bar g L\triangle\cap (R\ell\triangle+u+z))}{|\ell\triangle|}\\
 &=& \!\sum_{z\in\partial\cZ}\int_{[0,1]^3}du\int_{SO_3(\R)}dR \frac{E(\bar g L\triangle\cap (R\ell\triangle+u+z))}{|\ell\triangle|}\\
 & & \qquad\qquad+(\#\cZ)\int_{[0,1]^3}du\int_{SO_3(\R)}dR\ e_\ell(u,R).
\end{eqnarray*}
Using the stability property \textbf{(A2)}, we infer 
$$\frac{E(\bar g L\triangle\cap (R\ell\triangle+u+z))}{|\ell\triangle|}\geq -\kappa\frac{|\bar g L\triangle\cap (R\ell\triangle+u+z)|}{|\ell\triangle|}\geq-\kappa.$$ 
Hence
$$\int_G \frac{E(\bar g L\triangle\cap g\ell\triangle)}{|\ell\triangle|}dg \geq (\#\cZ)\int_{[0,1]^3\times SO_3(\R)}\ e_\ell(g)\,dg+\kappa(\#\partial\cZ).$$
As $\triangle$ has an $a$-regular boundary, it can be seen that $(\#\partial\cZ)\leq CL^2\ell$ and $\#\cZ=|L\triangle|+O(L^2\ell)$. Using again that $e_\ell$ is bounded, we eventually obtain the estimate
$$e_L(\bar g)\geq\int_{[0,1]^3\times SO_3(\R)}\ e_\ell(g)\,dg-C(\alpha(\ell)+\ell/L)$$
for some constant $C$. 
It is then an easy exercise to prove that 
$$\lim_{\ell\to\ii}\inf_{G}e_\ell=\lim_{\ell\to\ii}\int_{[0,1]^3\times SO_3(\R)}\ e_\ell:=\bar e$$ 
and finally that $e_\ell\to\bar e$ in $L^1([0,1]^3\times SO_3(\R))$.

The last step consists in proving the uniform convergence, using \textbf{(A4)}. Fix some small $\eta>0$. As $0\in\triangle$ and $\triangle$ is convex, we have $(1-\eta)\triangle\subset\triangle$. More precisely, there exists an $r>0$ and a neighborhood $W$ of the identity in $SO_3(\R)$ such that $R(1-\eta)\triangle+u\subset\triangle$ for all $(u,R)\in A:=B(0,r)\times W\subset G$. We have that $g\ell(1-\eta)\triangle\subset\ell\triangle$ for all $g\in A_\ell:=B(0,r\ell)\times W$, hence in particular for all $g\in A$. Now we fix some $\bar g\in G$ and apply \textbf{(A4)} with $\Omega=\bar g\ell\triangle$ and $\Omega'=\bar g g\ell(1-\eta)\triangle$, we get
$$E(\bar g\ell\triangle)\leq E(\bar g g\ell(1-\eta)\triangle)+C|\ell\triangle|\eta+o(|\ell\triangle|).$$
Integrating over $g\in A$ and dividing by $|\ell\triangle|$ we infer
$$e_\ell(\bar g)\leq \frac{1}{|\bar gA|}\int_{\bar g A}e_{(1-\eta)\ell}(g)\,dg+C\eta+o(1)_{\ell\to\ii}.$$
First we pass to the limit as $\ell\to\ii$ using that $e_\ell\to\bar e$ in $L^1(G)$ and $|A|\neq0$. Then we take $\eta\to0$ and get $\limsup_{\ell\to\ii}\sup_{\bar g\in G}e_\ell(\bar g)\leq \bar e$. This ends the proof of Theorem \ref{thm:limit_simplex}.
\qed

\bigskip

\noindent\textbf{Acknowledgement.} M.L. acknowledges support from the ANR project ``ACCQUAREL'' of the French ministry of research.

\addcontentsline{toc}{section}{References}
\bibliographystyle{amsplain}

\end{document}